\documentclass[sigconf]{aamas}

\usepackage{balance} 

\usepackage[colorinlistoftodos]{todonotes}
\usepackage{algorithm}
\usepackage{algorithmic}
\usepackage{subcaption}
\usepackage{enumitem}
\usepackage[most]{tcolorbox}
\usepackage{xcolor}
\usepackage{graphicx} 
\usepackage{tikz}
\usetikzlibrary{arrows.meta, positioning, shapes.geometric}

\doi{HWRI8919}

\makeatletter
\gdef\@copyrightpermission{
  \begin{minipage}{0.2\columnwidth}
   \href{https://creativecommons.org/licenses/by/4.0/}{\includegraphics[width=0.90\textwidth]{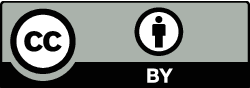}}
  \end{minipage}\hfill
  \begin{minipage}{0.8\columnwidth}
   \href{https://creativecommons.org/licenses/by/4.0/}{This work is licensed under a Creative Commons Attribution International 4.0 License.}
  \end{minipage}
  \vspace{5pt}
}
\makeatother

\setcopyright{ifaamas}
\acmConference[AAMAS '26]{Proc.\@ of the 25th International Conference
on Autonomous Agents and Multiagent Systems (AAMAS 2026)}{May 25 -- 29, 2026}
{Paphos, Cyprus}{C.~Amato, L.~Dennis, V.~Mascardi, J.~Thangarajah (eds.)}
\copyrightyear{2026}
\acmYear{2026}
\acmDOI{}
\acmPrice{}
\acmISBN{}

\title[Robust Information Design beyond Optimistic Equilibria]{Robust Information Design for Multi-Agent Systems with Complementarities: Smallest-Equilibrium Threshold Policies}

\author{Farzaneh Farhadi}
\affiliation{
  \institution{Aston University}
  \city{Birmingham}
  \country{United Kingdom}}
\email{f.farhadi@aston.ac.uk}

\author{Maria Chli}
\affiliation{
  \institution{Aston University}
  \city{Birmingham}
  \country{United Kingdom}}
\email{m.chli@aston.ac.uk}

\begin{abstract}
We study information design in multi-agent systems (MAS) with binary actions and strategic complementarities, where an external designer influences behavior only through signals. Agents play the \emph{smallest-equilibrium} of the induced Bayesian game, reflecting conservative, coordination-averse behavior typical in distributed systems. We show that when utilities admit a convex potential and welfare is convex, the robustly implementable optimum has a remarkably simple form: \emph{perfect coordination} at each state: either everyone acts or no one does. We provide a \emph{constructive threshold rule}: compute a one-dimensional score for each state, sort states, and pick a single threshold (with a knife-edge lottery for at most one state). This rule is an explicit optimal vertex of a linear program (LP) characterized by \emph{feasibility} and \emph{sequential obedience} constraints. Empirically, in both vaccination and technology-adoption domains, our constructive policy matches LP optima, scales as $O(|\Theta|\log|\Theta|)$, and avoids the inflated welfare predicted by \emph{obedience-only} designs that assume the designer can dictate the (best) equilibrium. The result is a general, scalable recipe for robust coordination in MAS with complementarities.
\end{abstract}


\newcommand{\BibTeX}{\rm B\kern-.05em{\sc i\kern-.025em b}\kern-.08em\TeX}

\begin{document}


\pagestyle{fancy}
\fancyhead{}


\maketitle 


\section{Introduction}
\paragraph{\textbf{Binary cooperation with complementarities.}}
Large-scale multi-agent systems (MAS) often involve binary cooperation decisions (such as whether to vaccinate \citep{bauch2004vaccination}, comply with regulations \citep{DELLANNO2009988}, or adopt a common standard \citep{katz1985network}), where incentives are complementary: the benefit of cooperating increases with the number of other cooperating agents. In such settings, a designer (e.g., public-health authority, regulator, platform operator) typically observes the environment more accurately than agents but cannot mandate actions, and thus must \emph{shape beliefs via information sharing}.

\paragraph{\textbf{Limitations of classical information design.}} Classical information design (e.g., Bayesian persuasion, Bayes--correlated equilibrium) studies \emph{partial implementation}: the designer specifies an information disclosure policy, makes recommendations accordingly (inducing a Bayesian game among agents), and assumes that when multiple equilibria exist, agents select an equilibrium consistent with those recommendations \citep{KamenicaGentzkow2011,BergemannMorris2016}. This optimistic assumption effectively allows the designer to select the \emph{best} equilibrium once incentive constraints hold. While analytically convenient, it can overstate what is achievable in decentralized or coordination-averse systems, where agents reason cautiously and coordination must emerge gradually rather than being taken for granted \citep{morris2024implementation,Chli2015}.

\paragraph{\textbf{Robust implementation via smallest-equilibrium play.}} We instead study implementation under smallest-equilibrium play, where agents coordinate on the worst equilibrium consistent with their information. This captures fragile coordination settings in which cooperation is chosen only if strictly profitable despite pessimistic beliefs.

\paragraph{\textbf{From LP characterization to constructive design.}} We characterize smallest-equilibrium implementability via a linear program (LP) with feasibility and sequential-obedience constraints. Under convex potential and welfare, the LP admits an explicit threshold solution, reducing design to scoring and sorting. The optimal policy achieves perfect coordination per state and is computable in \(O(|\Theta|\log|\Theta|)\) time.

\paragraph{\textbf{Empirical validation in MAS domains.}} We validate the framework in two contrasting MAS domains. \textbf{Case A (Vaccination)} uses a small discrete type space to make the mechanics of the threshold rule transparent. \textbf{Case B (Technology adoption)} approximates a continuous state space, showing scalability and generality. In both settings, our policy matches the LP optimum and remains \emph{robust under smallest-equilibrium play}, while classical partial-implementation designs predict systematically higher but non-credible welfare.

\paragraph{\textbf{Contributions.}}
\begin{itemize}\setlength\itemsep{2pt}
\item \textbf{From partial to robust implementation.} We contrast classical \emph{partial implementation}, where the designer can assume play of the best equilibrium, with the adversarial \emph{smallest-equilibrium} setting, and we formalize implementability through an LP with \emph{feasibility} and \emph{sequential obedience} constraints.
\item \textbf{Constructive, closed-form optimum.} Under convexity, the optimal smallest-equilibrium outcome is \emph{perfect coordination}, achieved by a \emph{threshold rule} based on a one-dimensional state score; this rule is an explicit optimal vertex of the LP.
\item \textbf{Scalable to large MAS.} The constructive rule reduces design to scoring and sorting, avoiding large-scale optimization while matching LP solutions and scaling well.
\item \textbf{Empirical evidence in MAS domains.} In vaccination and technology adoption, we quantify how partial-implementation baselines overstate welfare and show the robustness of our threshold policy.
\end{itemize}



\section{Related Work}

\paragraph{\textbf{Why not use existing methods? From classical to robust information design.}}
Classical information design, including \emph{Bayesian persuasion} and \emph{Bayes–correlated equilibrium (BCE)}, assumes that once obedience holds, agents coordinate on the designer’s preferred equilibrium \citep{KamenicaGentzkow2011,BergemannMorris2016}.  
Information disclosure has also been studied under this optimistic view in other settings, including security games \citet{Rabinovich2015}, Bayesian Stackelberg signaling \citet{Xu2016}, and persuasion with mechanism design \citet{Castiglioni2022}. 
These approaches rely on favorable equilibrium selection.

In many coordination problems, such optimism is implausible: cooperation must build gradually, and agents may anticipate deviation.  
This motivates \emph{robust information design}, which plans against \emph{worst-case / smallest-equilibrium selection}.  
Recent studies formalize implementability under such pessimistic play \citep{Mathevet2020,morris2024implementation} and develop robust disclosure for coordination, regime change, and panic-resilient systems \citep{InostrozaPavan2025,BasakZhou2025}. 
We adopt this robust perspective for \emph{supermodular MAS} via \emph{sequential obedience}, where agents expect only predecessors to comply.  
This strengthens implementability beyond simultaneous approaches and builds on dynamic and staged persuasion \citep{farhadi2018static,li2021sequential,farhadi2022dynamic,gan2022bayesian,yao2024dynamic}, aligning with recent work on \emph{constructive}, incentive-compatible, and computationally tractable mechanisms \citep{Farhadi2023JAIR}.

\paragraph{\textbf{Outcome- vs.\ policy-level implementability under smallest play.}}
Morris et al.\ \citep{morris2024implementation} provide necessary and sufficient \emph{outcome-level} conditions for distributions over actions and states achievable under smallest–equilibrium play. We translate these to a \emph{policy-level} characterization: an LP whose decision variables are conditional probabilities over \emph{sequences} of recommendations. Our LP comprises \emph{feasibility} and \emph{sequential obedience} constraints for cooperation and non-cooperation, which are the policy-wise counterparts of the outcome conditions. Beyond characterization, we extend~\citep{morris2024implementation} by introducing a \emph{constructive threshold rule policy}, which under convexity solves the LP in closed form and yields the exact optimal sequential design.

\paragraph{\textbf{MAS, network effects, and domain bridges.}}
Strategic \emph{complementarities} pervade many MAS domains, where an agent’s incentive to act increases with the number of others who act.  
In public health, vaccination and sanitation decisions create strong positive spillovers, making herd immunity a coordination challenge \citep{bauch2004vaccination,Guiteras2019Sanitation}.  
In \emph{public-good provision}, contributions exhibit increasing returns as participation grows, leading to multiple self-fulfilling equilibria and free-riding traps \citep{Andreoni1988}.  
\emph{Compliance} with taxation or regulation also depends on peers’ behavior, shaping incentives to comply \citep{allingham1972income,DELLANNO2009988}.  
Finally, \emph{technology adoption} and standardization are classic cases of network externalities, where adoption value rises with others’ uptake, creating tipping points and lock-in effects \citep{katz1985network,FarrellSaloner1985,Garlick2010agent}.

Our analysis focuses on the \emph{fully connected} benchmark, where complementarities are global. Extending to \emph{network games} with local complementarities is natural and discussed in the conclusion, as network topology can shift critical thresholds and alter the knife-edge mixing of robust policies.  
The constructive design we propose (i.e., \emph{score, sort, threshold}) is simple, interpretable, and well-suited to practical deployment in such structured MAS environments.

\section{Problem Setting: Binary-Action Complementarities in MAS Domains}\label{sec:problem-setting}

\begin{table*}[t]
\centering
\caption{Domain adapter: mapping model primitives to MAS applications.}
\label{tab:domain-adapter}
\begin{tabular}{p{3.2cm}p{3.5cm}p{3.2cm}p{3cm}p{2.5cm}}
\toprule
\textbf{Domain} & \textbf{States $\theta$} & \textbf{Cost $c$} & \textbf{Complementarity $\lambda_\theta$} & \textbf{Objective $V$} \\
\midrule
\textbf{Vaccination} 
& Disease prevalence; vaccine effectiveness  
& Health risks; monetary/time costs  
& Herd immunity (coverage $\uparrow$) 
& Social welfare from reduced infection \\
\addlinespace
\textbf{Tax / Fare Compliance} 
& Enforcement level (audit rate, penalties) 
& Payment of taxes or fares  
& Compliance spillovers (norms/deterrence) 
& Revenue net of enforcement \\
\addlinespace
\textbf{Public Good Provision} 
& Project productivity/value  
& Individual contribution cost (effort, money)  
& Thresholds / increasing returns 
& Surplus from funded project \\
\addlinespace
\textbf{Technology Adoption} 
& Market demand; regulatory push; ecosystem readiness 
& Migration, retraining, integration costs  
& Interoperability / network effects 
& Industry efficiency / consumer surplus \\
\bottomrule
\end{tabular}
\end{table*}

We study environments where each agent in a population chooses one of two actions: cooperation ($a_i = 1$) or non-cooperation ($a_i = 0$). The central feature is \emph{strategic complementarity}: an agent’s incentive to cooperate rises with the number of other cooperating agents. Such complementarities appear in many multi-agent systems: individuals deciding whether to vaccinate \citep{bauch2004vaccination}, taxpayers deciding whether to comply with rules or fares \citep{allingham1972income}, communities deciding whether to contribute to a public project \citep{Andreoni1988}, or firms deciding whether to adopt a technology standard \citep{katz1985network,FarrellSaloner1985}. In each case, the information designer is an authority (e.g., a regulator, public health agency, tax administration, or standards consortium) that can influence outcomes only through information disclosure.

Our formulation abstracts away from domain specifics to capture three common primitives: a binary individual choice, payoff complementarities, and an external designer with superior information. We introduce a \emph{domain adapter} (Table~\ref{tab:domain-adapter}) that maps these primitives to diverse MAS applications, illustrating the scope of our analysis.

\subsection{Environment}

\textbf{Agents.}
There are $N$ agents, indexed by $i = 1, \dots, N$. Each chooses either the costly cooperative action ($a_i = 1$) or the default non-cooperative action ($a_i = 0$). An \emph{action profile} is $a = (a_1,\dots,a_N)$, and $n(a) = \sum_{i=1}^N a_i$ denotes the total number of cooperating agents.

\textbf{State.}
The payoff environment depends on an underlying state $\theta \in \Theta$, drawn from a commonly known prior distribution $\mu$ over the finite set $\Theta$. This state captures exogenous conditions (such as infection prevalence, audit intensity, regulatory pressure, or market demand) that determine both the intrinsic value of cooperation and the strength of complementarities. Agents know only the prior but cannot observe the realized state.

\textbf{Designer.}
Unlike the agents, a \emph{designer} observes state $\theta$ and can send signals to influence the agents' choices. The designer’s objective (formalized in Section~\ref{sec:designer}) differs from the agents’ individual payoffs and seeks to steer the system toward desirable collective outcomes. This information asymmetry underpins the designer’s role: by shaping the flow of information, the designer can coordinate agents’ behavior and improve system-level performance.

\subsection{Agents' utilities and incentives}\label{sec:utilities}

\textbf{Utility function.} Agent $i$’s utility in action profile $a$ and state $\theta$ is
\begin{equation}
u_i(a,\theta)\;=\;a_i\Big(b_\theta+\lambda_\theta\,\frac{n(a_{-i})}{N-1}\Big)\;-\;c\,a_i\;+\;\kappa_i(a_{-i},\theta),
\label{eq:utility}
\end{equation}
where $a_{-i}$ denotes the action profile of all agents except $i$, and $n(a_{-i})=\sum_{j\neq i} a_j$ is the number of other agents who choose the cooperative action.

\textbf{Components.} This specification has four essential components. First, $b_\theta$ is the baseline return to cooperation in state $\theta$ (e.g., private vaccination benefits or intrinsic gains from technology adoption). Second, the term $\lambda_\theta \frac{n(a_{-i})}{N-1}$ with $\lambda_\theta \ge 0$ captures \emph{complementarities}: the payoff from cooperation increases with others’ participation. Third, cooperation incurs a fixed cost $c>0$ (e.g., medical or migration costs). Fourth, $\kappa_i(a_{-i},\theta)$ captures agent heterogeneity, such as legacy revenues, subsidies, or exposure to local risks.

\smallskip
\textbf{Assumption (Dominance).}
We assume the existence of a \emph{dominant state} $\theta^{\mathrm{dom}} \in \Theta$ such that
\[
b_{\theta^{\mathrm{dom}}} - c \;>\; 0.
\]
In this state, cooperation $a_i=1$ is strictly profitable for every agent regardless of others’ actions, ensuring that universal cooperation is feasible in some environments.

\textbf{Strategic complementarity.} From \eqref{eq:utility}, each agent $i$'s \emph{marginal gain} from cooperation is
\begin{equation}
G_i(a_{-i},\theta) \triangleq u_i((1,a_{-i}),\theta)-u_i((0,a_{-i}),\theta)=b_\theta-c+\lambda_\theta\frac{n(a_{-i})}{N-1}.
\label{eq:marginal-gain}
\end{equation}
The marginal gain is increasing in $n(a_{-i})$, establishing \emph{strategic complementarity}. Moreover, the symmetry of marginal gains across agents implies the existence of a potential function.

\subsection{Potential representation and convexity}

A key property of the game defined in Section~\ref{sec:utilities} is that it admits a \emph{potential function}.  
A potential function $\Phi$ is a single scalar that tracks unilateral incentives: the change in any agent’s utility from flipping its action equals the corresponding change in $\Phi$.  
Best responses therefore \emph{climb} the potential, and many equilibrium and comparative-statics results follow directly from its shape.

In our setting, the potential takes the form
\begin{equation}\label{eq:potential}
\Phi(a,\theta)=F_\theta\bigl(n(a)\bigr), 
\qquad 
F_\theta(n)=(b_\theta-c)\,n+\lambda_\theta\,\frac{n(n-1)}{2(N-1)}.
\end{equation}
Indeed,
\[
\Phi((1,a_{-i}),\theta)-\Phi((0,a_{-i}),\theta)
  = b_\theta - c + \lambda_\theta\frac{n(a_{-i})}{N-1}, \quad \forall i,
\]
which coincides exactly with the marginal gain in~\eqref{eq:marginal-gain}.  
Thus the game is an \emph{exact potential game} \citep{monderer1996potential}: any unilateral deviation changes an agent’s utility by the same amount as it changes $\Phi$.

\begin{lemma}[Convexity of the potential]
For each state $\theta \in \Theta$, the potential function $F_\theta(n)$ defined in~\eqref{eq:potential} is (discretely) convex in $n$.
\end{lemma}

\begin{proof}
The discrete increments
\[
\Delta F_\theta(k)=F_\theta(k+1)-F_\theta(k)
   =(b_\theta-c)+\lambda_\theta\frac{k}{N-1}
\]
are weakly increasing in $k$ because $\lambda_\theta \ge 0$.  
Hence $F_\theta$ is discretely convex.
\end{proof}

\subsection{Timing and information}\label{sec:timing}
The interaction unfolds in three stages. (i) The designer \emph{commits} to an information disclosure policy $\pi$, which maps each possible state to a distribution of signals. These signals can be public or private and may be delivered either simultaneously or sequentially to the agents.
(ii) Nature draws the state $\theta\sim\mu$.  
(iii) Signals are generated according to the committed policy; upon receiving them, agents update their beliefs and play a Bayesian Nash equilibrium (NE) of the resulting incomplete-information game.

Because the game exhibits \emph{strategic complementarities}, best responses are monotone and the set of Bayesian NE can be ordered by the number of agents who choose the costly action. When multiple equilibria exist, we assume that agents coordinate on the one with the fewest cooperators, referred to as the \emph{smallest-equilibrium}. Standard lattice results guarantee that this equilibrium exists and can be reached by iterated best responses starting from universal inaction \cite{Topkis1979,milgrom1990rationalizability,Vives2005}. This selection models coordination-averse systems: agents take the costly cooperative action only when it is strictly beneficial, which makes the resulting outcomes robust to coordination failure and to pessimistic interpretations of signals.

\subsection{Designer’s problem}\label{sec:designer}

The designer evaluates collective performance through a welfare function 
\(V: A \times \Theta \to \mathbb{R}\),
normalized so that \(V(\mathbf{0},\theta)=0\) for all \(\theta\), where \(\mathbf{0}\) denotes universal non-cooperation.  
We assume that welfare is weakly increasing in the number of cooperating agents and satisfies the following convexity property.

\smallskip
\noindent\textbf{Assumption (Convex welfare).}
For each \(\theta \in \Theta\),
\[
V(a,\theta)\;\leq\;\frac{n(a)}{N}\,V(\mathbf{1},\theta),
\]
where \(\mathbf{1}\) denotes universal cooperation.  
This assumption captures the idea that the marginal social benefit of additional cooperating agents does not decrease as participation grows.

\smallskip
The designer’s goal is to choose an information disclosure policy $\pi$ that maximizes expected welfare,
\begin{equation}\label{eq:exp-welfare}
\max_{\pi}\;\mathbb{E}_{\theta \sim \mu}\bigl[V(a^{\ast}_{\pi}(\theta),\theta)\bigr],
\end{equation}
where \(a^{\ast}_{\pi}(\theta)\) is the smallest Bayesian NE induced by the policy $\pi$ in state \(\theta\).  
The designer has no direct control over agents’ actions and can only shape their beliefs through the signals provided.

Solving the designer’s problem directly is challenging because the space of possible signaling schemes is vast: a policy may assign arbitrary public or private signals, possibly revealed in sequence, and induces complex belief updates and equilibria.  
To proceed, we next characterize the structure of information policies that are sufficient for optimality and derive the incentive constraints they must satisfy.  
This leads to a tractable \emph{linear program}, which forms the foundation for our constructive solution in Section~\ref{sec:perfect-coordination}.

\section{Information Design and Implementability}
\label{sec:info-design}

Classical information design makes an optimistic assumption: if several equilibria exist, agents coordinate on the one the designer prefers. With this assumption, the designer can restrict attention to \emph{simultaneous recommendation policies}, where all agents receive private recommendations at once, and need only ensure that each recommendation is individually optimal to follow (obedience).

Our setting departs from this view by requiring robustness to the \emph{smallest-equilibrium}: agents act conservatively and reason under the most cautious belief about how many others will cooperate, rather than optimistically believing that everyone will follow recommendations. In such environments, simultaneous recommendations are generally insufficient.  
The designer may need to disclose information \emph{sequentially}, allowing cooperation to build gradually while remaining incentive compatible under these conservative beliefs.  
It has been shown in~\cite{morris2024implementation} that this conservative reasoning can be modeled by assuming that each agent expects only those recommended \emph{before} them to comply with their recommendations. Ensuring that agents still follow their recommendations under this belief leads to a strengthened incentive condition, which we call \emph{sequential obedience}, a stricter form of the classical obedience constraint.

In what follows, we formally define sequential information policies, introduce the sequential obedience condition, and show how these elements yield a linear programming formulation of the designer’s problem.

\subsection{Sequential information policies}

A \emph{sequential policy} discloses recommendations in stages (e.g., staggered rollouts where only some agents are informed first).  
We model each staged recommendation process by an \emph{ordered sequence} 
\(\gamma = (i_1,\ldots,i_m)\) of distinct agents, representing the order in which they are invited to cooperate (receive the recommendation $a_i=1$).  
Agents not appearing in \(\gamma\) receive the default non-cooperation recommendation ($a_i=0$).  
Let \(\Gamma\) denote the set of all such ordered subsequences of $N$ agents, including the empty sequence \(\varnothing\).

A \emph{sequential information policy} is then a family of conditional distributions
\[
\pi^{\mathrm{seq}}(\gamma \mid \theta), \qquad
\gamma \in \Gamma, \;\theta \in \Theta,
\]
where \(\pi^{\mathrm{seq}}(\gamma \mid \theta)\) is the probability that the designer, upon observing state \(\theta\), runs the staged recommendation process that invites agents exactly in the order specified by \(\gamma\).
For each \(\theta\), the policy must define a valid probability distribution:
\[
\pi^{\mathrm{seq}}(\gamma \mid \theta)\ge 0
\quad\text{and}\quad
\sum_{\gamma\in\Gamma}\pi^{\mathrm{seq}}(\gamma \mid \theta)=1,
\quad\forall\,\theta\in\Theta.
\tag{Feas}
\]
These requirements are the \emph{feasibility constraints}.

\paragraph{\textbf{Example 1.}}
Consider a setting with $N=3$ agents $\{1,2,3\}$ and a state space $\Theta=\{L,H\}$.  
A possible sequential policy might be
\[
\pi^{\mathrm{seq}}((1,3)\mid L)=0.6,\quad
\pi^{\mathrm{seq}}((2,3)\mid L)=0.4,
\]
\[
\pi^{\mathrm{seq}}((3,1,2)\mid H)=1,
\]
meaning that in state \(L\) the designer invites agents to cooperate in the order \((1,3)\) with probability \(0.6\) and \((2,3)\) with probability \(0.4\), while in state \(H\) the designer always sends cooperative invitations in the order \((3,1,2)\) (first agent \(3\), then agent \(1\), then agent \(2\)). \qed

\subsection{Sequential obedience}

Having introduced sequential policies, we now strengthen the classical incentive requirement to capture the conservative reasoning of smallest-equilibrium play.  
Under this belief system, an invited agent acts as if only those invited \emph{before} it will comply.  
For an information policy to be feasible, every agent must find it optimal to follow its recommendation, whether to cooperate or not, even under this most cautious belief. We refer to this requirement as \emph{sequential obedience}.

Let $a^{\mathrm{pred}}_{-i}(\gamma)$ denote the action profile in which each $j\neq i$ cooperates if and only if $j$ appears \emph{before} $i$ in the sequence $\gamma$.  
(If $i\notin\gamma$, then $a^{\mathrm{pred}}_{-i}(\gamma)$ simply has all agents in $\gamma$ cooperating and all others not.)  
Using the marginal gain function $G_i(a_{-i},\theta)$ defined in~\eqref{eq:marginal-gain}, we express sequential obedience as two families of linear constraints.

\begin{definition}[Sequential obedience — cooperation]
\label{def:SO-C}
A sequential policy $\pi^{\mathrm{seq}}$ satisfies \emph{sequential obedience for cooperation} if, for every agent $i$,
\[
\sum_{\theta\in\Theta}\mu(\theta)
\sum_{\gamma\in\Gamma_i}
\pi^{\mathrm{seq}}(\gamma \mid \theta)\,
G_i\bigl(a^{\mathrm{pred}}_{-i}(\gamma),\theta\bigr)x
\;\ge 0,
\tag{SO-C}
\]
where $\Gamma_i=\{\gamma\in\Gamma : i\in\gamma\}$ is the set of sequences in which $i$ is invited to cooperate.
\end{definition}

Constraint~(SO-C) requires that even a fully cautious agent, believing only earlier invitees will cooperate, still finds it individually rational to accept a recommendation to cooperate. The inequality in~(SO-C) is weak because sequential obedience allows indifference: strict gains identify the smallest equilibrium, while obedience requires only a weak best response.

\begin{definition}[Sequential obedience — non-cooperation]
\label{def:SO-N}
A sequential policy $\pi^{\mathrm{seq}}$ satisfies \emph{sequential obedience for non-cooperation} if, for every agent $i$,
\[
\sum_{\theta\in\Theta}\mu(\theta)
\sum_{\gamma\in\Gamma\setminus\Gamma_i}
\pi^{\mathrm{seq}}(\gamma \mid \theta)\,
G_i\bigl(a^{\mathrm{pred}}_{-i}(\gamma),\theta\bigr)
\;\le 0,
\tag{SO-N}
\]
where $\Gamma\setminus\Gamma_i=\{\gamma\in\Gamma : i\notin\gamma\}$ is the set of sequences in which $i$ is \emph{not} invited to cooperate.
\end{definition}

Constraint~(SO-N) states that an agent who is not invited must not have an incentive to deviate and cooperate.

\paragraph{\textbf{Example 2 (Checking sequential obedience).}}
Consider Example~1 and suppose the primitives are $\mu(L)=\mu(H)=0.5, b_L=1, b_H=2.4, \lambda_L=0.1, \lambda_H=0.5, c=2$. We illustrate the check for sequential obedience for cooperation (SO-C) using player~3.

\medskip
\noindent\emph{Player 3.} Sequences that invite player~3 are $(1,3)$ and $(2,3)$ in state $L$, and $(3,1,2)$ in state $H$. The marginal gains of cooperation in these cases are
\[
G_3(a^{\mathrm{pred}}_{-3}((1,3)),L)=G_3(a^{\mathrm{pred}}_{-3}((2,3)),L)
= b_L - c + \lambda_L \tfrac{1}{2}
= -0.95,
\]
\[
G_3(a^{\mathrm{pred}}_{-3}((3,1,2)),H)
= b_H - c + \lambda_H \tfrac{0}{2}
= 0.4.
\]
Sequential obedience for cooperation (SO-C) for player~3 is
\[
\sum_{\theta\in\{L,H\}}\mu(\theta)
\sum_{\gamma\in\Gamma_3}\pi^{\mathrm{seq}}(\gamma\mid\theta)\,
G_3\bigl(a^{\mathrm{pred}}_{-3}(\gamma),\theta\bigr)\;\ge 0,
\]
which evaluates to
\[
0.5\Bigl[\,0.6 * (-0.95) + 0.4*(-0.95)\Bigr]
\;+\;
0.5\Bigl[\,1.0* 0.4\Bigr]
=
-0.275
\;<\; 0.
\]
Thus (SO-C) fails for player~3, and the policy in Example~1 is not sequentially obedient. \qed

Sequential obedience is strictly stronger than the standard (simultaneous) obedience used in Bayesian persuasion and BCE.  Standard obedience requires that recommended actions be optimal assuming all other agents follow their recommendations. In contrast, sequential obedience (SO-C and SO-N) strengthens these constraints by certifying that each recommendation is a best reply under the conservative ``only predecessors comply'' belief.  
A simple one-state example illustrates the gap.

\paragraph{\textbf{Example 3 (Sequential obedience is stronger than ordinary obedience).}}
Consider $N=3$ agents and a single state $\Theta=\{K\}$ with primitives
$(b_K,\lambda_K,c)=(1,1.5,2)$. Suppose the designer commits to a policy that always invites all agents to cooperate in the fixed order \(\gamma = (1,2,3)\).

\emph{Standard obedience.}
If every other agent cooperates, one agent’s deviation gain is $G_i(\mathbf{1},K)=b_K-c+\lambda_K=0.5>0$, so “all cooperate’’ is a Bayesian NE and satisfies the usual obedience constraints.

\emph{Sequential obedience.}
Under smallest–equilibrium reasoning, an invited agent expects only earlier invitees to cooperate.  
The first invitee faces \(G_1(\mathbf{0},K)=b_K-c=-1<0\) and would deviate, breaking the cooperation chain. Hence, the policy violates (SO-C) at the very first invitation. \qed

This simple case shows that a recommendation profile implementable under standard/BCE obedience may fail under the stricter sequential obedience required for smallest–equilibrium play.

\subsection{Characterizing implementability and the designer’s problem}

Combining these conditions yields the following.

\begin{theorem}[Characterization of implementable sequential policies]
\label{thm:S-implementability}
A sequential policy $\pi^{\mathrm{seq}}$ induces play consistent with smallest-equilibrium behavior if and only if it satisfies:
\begin{enumerate}
    \item \textbf{Feasibility} (Feas),
    \item \textbf{Sequential obedience — cooperation} (SO-C),
    \item \textbf{Sequential obedience — non-cooperation} (SO-N).
\end{enumerate}
\end{theorem}

\begin{proof}[Proof sketch]
\cite{morris2024implementation} characterize, at the \emph{outcome level}, the distributions over action profiles and states that are implementable under smallest-equilibrium play. Our conditions (Feas), (SO-C), and (SO-N) provide the corresponding \emph{policy-level} characterization, implying that any smallest-equilibrium outcome achievable by an arbitrary information disclosure scheme can also be implemented by a sequential policy satisfying these constraints.
\end{proof}

\medskip
These three families of constraints are linear in the policy distributions $\pi^{\mathrm{seq}}(\gamma \mid \theta)$. 
The designer's optimization problem can therefore be written as the following linear program:
\begin{equation}
\label{eq:optimization-seq}
\begin{aligned}
\max_{\pi^{\mathrm{seq}}}\quad &
   \sum_{\theta\in\Theta} \mu(\theta)
   \sum_{\gamma\in\Gamma}
      \pi^{\mathrm{seq}}(\gamma \mid \theta)
      \,V\bigl(a(\gamma),\theta\bigr), \\
\text{s.t.}\quad 
& \text{(Feas)} , \text{(SO-C)}, \text{(SO-N)},
\end{aligned}
\end{equation}
where $a(\gamma)$ is the action profile in which each player $i$ plays $1$ if and only if $i$ is listed in $\gamma$. The objective function in~\eqref{eq:optimization-seq} is exactly the designer’s expected welfare objective introduced in~\eqref{eq:exp-welfare}, expressed here in terms of the sequential policy variables.

Every optimal solution of~\eqref{eq:optimization-seq} corresponds to a sequential information policy that is robust under smallest-equilibrium play.  
Although problem~\eqref{eq:optimization-seq} is a \emph{linear program} and can, in principle, be solved using standard LP solvers, its size grows extremely quickly.  
The number of decision variables (i.e., the probabilities \(\pi^{\mathrm{seq}}(\gamma \mid \theta)\) for all \(\gamma \in \Gamma\) and \(\theta \in \Theta\)) is \(|\Theta| \cdot |\Gamma|\), where
\[
|\Gamma| = \sum_{k=0}^{N} \binom{N}{k} k! = O(N!),
\]
denotes the number of ordered subsequences of the \(N\) agents.  
Since \(N!\) grows super-exponentially with \(N\), the problem becomes computationally infeasible even for moderate \(N\).  
To overcome this combinatorial blow-up, the next section introduces a \emph{constructive algorithm} that computes an \emph{exact optimal sequential policy} for~\eqref{eq:optimization-seq} in time \(O(|\Theta| \log |\Theta|)\), independent of the number of agents \(N\).

\section{Perfect Coordination under Convexity: A Constructive Characterization}
\label{sec:perfect-coordination}

Section~\ref{sec:info-design} characterized the feasible set of \emph{sequential information policies} via three families of linear constraints (feasibility, sequential obedience for cooperation, and sequential obedience for non-cooperation) and showed that the designer’s problem reduces to the linear program~\eqref{eq:optimization-seq}.  
A natural question for MAS is whether the optimal robust policy must be combinatorially complex (e.g., mixing over many partial invitation sequences) or whether a simple, statewise structure suffices.

The answer is strikingly simple: there always exists an optimal smallest-equilibrium implementable sequential policy that achieves perfect coordination in every state: either all agents cooperate or none do. Convexity of the welfare function removes any advantage of partial adoption, while sequential obedience then enforces a sharp, statewise ``all-or-none'' rule that remains robust under the most cautious agent beliefs.

\subsection{A constructive threshold rule}

Our approach is constructive: we explicitly \emph{build} an optimal perfectly coordinated policy \(\pi^{\mathrm{seq}\,\star}\) and then verify its optimality.

\paragraph{Step 1: Score each state.}  
For each state $\theta$, define
\[
S(\theta) \;=\; \frac{F_\theta(N)}{V(\mathbf{1},\theta)},
\]
the ratio of the potential to welfare under full cooperation.

\paragraph{Step 2: Order states and identify a threshold.}  
Order states so that $S(\theta)$ is weakly increasing.\footnote{Any tie-breaking among states with the same $S(\theta)$ does not affect the construction.} We write $\theta' \succ \theta$ if $\theta'$ comes after $\theta$ in that ordering.  

Let $\theta^\star$ be the (unique) \emph{threshold} state and $p^\star\in[0,1]$ the mixing weight that solve
\begin{equation}\label{eq:p-star}
\sum_{\theta:\,\theta \succ \theta^\star}\mu(\theta)F_\theta(N)
\;+\;
p^\star \mu(\theta^\star) F_{\theta^\star}(N)
\;=\; 0,
\end{equation}
i.e., the weighted sum of full-cooperation potentials across the ``cooperate'' states balances exactly at zero.

\paragraph{Step 3: Construct the perfectly coordinated policy.}  
Define the sequential policy $\pi^{\mathrm{seq}\,\star}$ by assigning probability mass only to the two extreme invitation sets: the empty sequence (no one invited) and the full sequence (all invited).  
For any ordered sequence $\gamma\in\Gamma$, let $|\gamma|$ denote its length.  Then
\begin{equation}
\pi^{\mathrm{seq}\,\star}(\gamma\mid\theta)
=
\begin{cases}
1/N!, & \!|\gamma|=N,\;\theta\succ\theta^\star,\\[2pt]
p^\star/N!, & \!|\gamma|=N,\;\theta=\theta^\star,\\[2pt]
1-p^\star, & \!|\gamma|=0,\;\theta=\theta^\star,\\[2pt]
1, & \!|\gamma|=0,\;\theta\prec\theta^\star,\\[2pt]
0, & \!\text{otherwise.}
\end{cases}
\label{eq:pi-star}
\end{equation}
States above the threshold induce \emph{universal cooperation} (any full ordering of agents chosen uniformly); states below induce \emph{universal defection}; and the threshold state mixes between them in just the right proportion.  
This construction balances the incentives implied by sequential obedience: cooperation is sustained only where the environment is strong enough, while weaker states lead to coordinated defection. The optimal policy $\pi^{\mathrm{seq}\,\star}$ sets the (SO-C) constraints exactly at equality via the choice of $p^\star$~\eqref{eq:p-star}.  
An algorithmic description appears in Algorithm~\ref{alg:optimal-policy}.

\begin{algorithm}[t]
\caption{Constructing the perfectly coordinated optimal sequential policy $\pi^{\mathrm{seq}\,\star}$}
\label{alg:optimal-policy}
\begin{algorithmic}[1]
\STATE Compute $S(\theta)=F_\theta(N)/V(\mathbf{1},\theta)$ for all $\theta$.
\STATE Sort states by nondecreasing $S(\theta)$.
\STATE Find the smallest $\theta^\star$ such that
\[
\sum_{\theta:\,\theta \succ \theta^\star}\mu(\theta)F_\theta(N)
\;\le\;0
\;<\;
\sum_{\theta:\,\theta \succeq \theta^\star}\mu(\theta)F_\theta(N).
\]
\STATE Choose $p^\star$ so that
\[
\sum_{\theta:\,\theta \succ \theta^\star}\mu(\theta)F_\theta(N)
\;+\;
p^\star \mu(\theta^\star) F_{\theta^\star}(N)
=0.
\]
\STATE Build $\pi^{\mathrm{seq}\,\star}$ according to \eqref{eq:pi-star}.
\end{algorithmic}
\end{algorithm}

\subsection{Feasibility of $\pi^{\mathrm{seq}\,\star}$}

\begin{lemma}[Feasibility]
\label{lem:pi-star-feasible}
The sequential policy $\pi^{\mathrm{seq}\,\star}$ in~\eqref{eq:pi-star} satisfies feasibility (Feas), sequential obedience for cooperation (SO-C), and sequential obedience for non-cooperation (SO-N).
\end{lemma}

\begin{proof}[Proof sketch]
Feasibility follows from constructing valid distributions in each state. 
The choice of \(p^\star\) in~\eqref{eq:p-star}, together with the symmetric construction in~\eqref{eq:pi-star}, ensures the (SO-C) constraints hold with equality for all agents. 
Moreover, convexity of the potential implies that in any state where full cooperation is not recommended, no single agent can profitably deviate; hence (SO-N) is satisfied.
\end{proof}

\subsection{Optimality}

\begin{theorem}[Perfect coordination is optimal]
\label{thm:perfect}
In any binary-action supermodular game with a \emph{convex} potential and a designer whose welfare is \emph{convex}, the sequential policy $\pi^{\mathrm{seq}\,\star}$ defined in~\eqref{eq:pi-star} is an optimal solution of the linear program~\eqref{eq:optimization-seq}.
\end{theorem}

\begin{proof}[Proof sketch]
The LP \eqref{eq:optimization-seq} is a convex optimization over a polytope.  
Convexity of both the potential and welfare ensures that extreme points of this feasible set are all-or-none policies.  
Balancing sequential obedience identifies the optimal extreme point, which coincides with $\pi^{\mathrm{seq}\,\star}$.
\end{proof}

\subsubsection*{\textbf{Generality and Computational Efficiency}}
Theorem~\ref{thm:perfect} relies on our standing assumptions: an exact-potential game and convexity of the potential and welfare.  
In contrast, the LP~\eqref{eq:optimization-seq} applies to any binary-action supermodular game without these assumptions but becomes computationally intractable as the number of agents grows.  
Our constructive method instead yields a closed-form optimum under convexity and is computationally efficient: computing state scores requires \(O(|\Theta|)\) time, sorting them takes \(O(|\Theta|\log|\Theta|)\), and finding the optimal threshold and mixing weight is \(O(|\Theta|)\).  
Overall, the entire policy can be computed in \(O(|\Theta|\log|\Theta|)\) time, independent of the population size \(N\).

\subsection{MAS takeaways and implementation guidance}

\begin{figure*}[t]
\centering
\begin{tikzpicture}[
    font=\small,
    panel/.style={draw, rounded corners=4pt, minimum width=4cm, minimum height=4.3cm, align=left},
    title/.style={font=\bfseries\scriptsize},
    blob/.style={fill opacity=0.35, draw opacity=0},
    blobline/.style={draw=black, opacity=0.10, line width=0.4pt},
    staropt/.style={star, star points=5, star point ratio=2.2, minimum size=1mm, draw=black, fill=green!60!black},
    arrowflow/.style={-Latex, thick},
    algo/.style={draw, rounded corners=4pt, fill=gray!10, minimum height=1.1cm, align=center},
    runtime/.style={draw, rounded corners=3pt, fill=white, inner sep=4pt},
    axis/.style={draw=black, opacity=0.16, line width=0.4pt}
]

\node[panel, fill=blue!5] (A) {};
\node[panel, fill=orange!6, right=1.5cm of A] (B) {};
\node[panel, fill=green!6, right=1.5cm of B] (C) {};

\node[title, anchor=north] at ([yshift=-1mm]A.north)
{Classical persuasion / BCE};

\node[title, anchor=north] at ([yshift=-1mm]B.north)
{Robust Implementability};

\node[title, anchor=north] at ([yshift=-1mm]C.north)
{Convex potential and welfare};

\begin{scope}[shift={(A.center)}, yshift=4mm]
  \draw[axis] (-1.8,0) -- (1.8,0);
  \draw[axis] (0,-1.6) -- (0,1.3);

  \fill[blue!60, blob]
      (-1.45,-0.75) .. controls (-1.65,0.65) and (-0.65,1.45)
      .. (0.65,1.05)
      .. controls (1.65,0.55) and (1.35,-0.85)
      .. (0.25,-1.5)
      .. controls (-0.85,-1.45) and (-1.45,-0.75)
      .. cycle;

  \node[circle, fill=blue!80!black, minimum size=5pt, inner sep=0pt] at (0.15,0.9) {};
\end{scope}

\node[anchor=south west, font=\scriptsize] at ([xshift=2mm,yshift=1mm]A.south west)
{\begin{tabular}{l}
$\bullet$ Obedience only \\
$\bullet$ Optimistic selection
\end{tabular}};

\begin{scope}[shift={(B.center)}, yshift=4mm]
  \draw[axis] (-1.8,0) -- (1.8,0);
  \draw[axis] (0,-1.6) -- (0,1.3);

  \fill[blue!60, fill opacity=0.18, draw opacity=0]
      (-1.45,-0.75) .. controls (-1.65,0.65) and (-0.65,1.45)
      .. (0.65,1.05)
      .. controls (1.65,0.55) and (1.35,-0.85)
      .. (0.25,-1.5)
      .. controls (-0.85,-1.45) and (-1.45,-0.75)
      .. cycle;


  \begin{scope}[yshift=-2mm]
  \fill[orange!80!black, blob]
      (-1.05,-0.55) .. controls (-1.15,0.55) and (-0.35,1.10)
      .. (0.70,0.88)
      .. controls (1.20,0.30) and (0.95,-0.70)
      .. (0.25,-0.98)
      .. controls (-0.65,-1.08) and (-1.05,-0.55)
      .. cycle;
      \node[circle, fill=orange!80!black, minimum size=5pt, inner sep=0pt] at (0.42,0.72) {};
  \end{scope}
\end{scope}

\node[anchor=south west, font=\scriptsize] at ([xshift=2mm,yshift=0mm]B.south west)
{\begin{tabular}{l}
$\bullet$ Sequential obedience \\
$\bullet$ Conservative beliefs \\
$\bullet$ Robust feasible set shrinks
\end{tabular}};

\begin{scope}[shift={(C.center)}, yshift=4mm]
  \draw[axis] (-1.8,0) -- (1.8,0);
  \draw[axis] (0,-1.6) -- (0,1.3);



  \begin{scope}[yshift=-2mm]
  \fill[orange!80!black, fill opacity=0.3, draw opacity=0]
      (-1.05,-0.55) .. controls (-1.15,0.55) and (-0.35,1.10)
      .. (0.70,0.88)
      .. controls (1.20,0.30) and (0.95,-0.70)
      .. (0.25,-0.98)
      .. controls (-0.65,-1.08) and (-1.05,-0.55)
      .. cycle;
\end{scope}

  \draw[fill=green!35!white]
      (-0.98,-0.8) --
      (-0.58,0.47) --
      (0.89,0.32) --
      (0.715,-0.80) --
      cycle;

  \node[staropt, scale=0.3] at (0.89,0.32) {};

\end{scope}

\node[anchor=south west, font=\scriptsize] at ([xshift=2mm,yshift=0mm]C.south west)
{\begin{tabular}{l}
$\bullet$ Optimum at an extreme point \\
$\bullet$ Perfect coordination: all or none \\
$\bullet$ Threshold on score $S(\theta)$
\end{tabular}};

\draw[arrowflow] (A.east) --
    node[above, font=\scriptsize]{Robustness}
    node[below, font=\scriptsize]{requirement}
(B.west);

\draw[arrowflow] (B.east) -- node[above, font=\scriptsize]
{Convexity} (C.west);

\node[algo, below=0.35cm of B, minimum width=14cm] (ALG) {};

\node[align=center] at (ALG.center)
{
\textbf{Constructive optimal policy}
\\[3pt]
Compute $S(\theta)$ $\rightarrow$ sort $\theta$ $\rightarrow$ threshold $\theta^\star$ $\rightarrow$ mix at one state
};

\node[runtime, anchor=east, align=center, font=\scriptsize]
at ([xshift=-4mm, yshift=0mm]ALG.east)
{
Runtime: $\mathcal{O}(|\Theta|\log|\Theta|)$ \\
Independent of $N$
};
\end{tikzpicture}

\caption{
Robust information design under smallest-equilibrium play. Sequential obedience restricts implementability relative to classical persuasion. Under convex potential and welfare, the optimum is an extreme point, yielding a threshold policy computable in $\mathcal{O}(|\Theta|\log|\Theta|)$ time.
}
\label{fig:visualization}
\end{figure*}
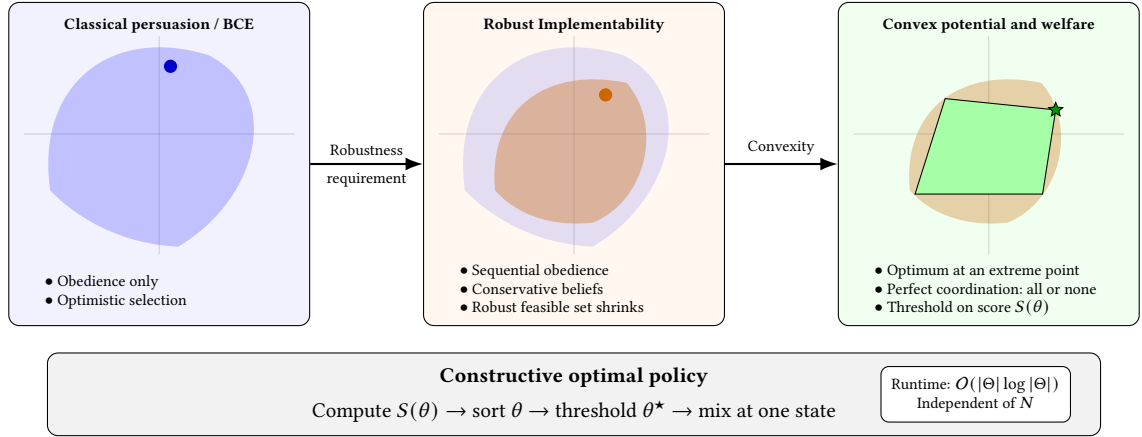

\setlist[itemize]{left=0pt}
\begin{itemize}
\item \textbf{One-dimensional score.} States are ranked by 
$S(\theta)=F_\theta(N)/\allowbreak V(\mathbf{1},\theta)$; sorting and balancing sequential obedience yields the exact optimum.

\item \textbf{All-or-none robustness.} Under convexity, optimal robust policies never use partial adoption within a state: either \emph{everyone cooperates} or \emph{no one does}.

\item \textbf{Closed-form scalability.} The threshold rule is an explicit LP vertex computable in $\mathcal{O}(|\Theta|\log|\Theta|)$ time.
\end{itemize}

Figure \ref{fig:visualization} visually summarizes the transition from classical persuasion to robust implementability and the resulting threshold structure.



\section{Numerical Results}
\label{sec:results}

We validate the constructive characterization of Section~\ref{sec:perfect-coordination} through two multi-agent domains.  
The first revisits the running example to clarify the threshold mechanics and highlight the contrast with partial implementation.  
The second considers a nearly continuous state space, showing that the approach remains tractable and interpretable beyond small illustrative settings.  
All results are reproducible; code is available~\citep{Farhadi2026Code}.

\subsection{Case study~1: Vaccination-style environment}

\paragraph{Setup.} We revisit the running example as a vaccination scenario with $N=3$ agents choosing $a_i\in\{0,1\}$.  
States are $\Theta=\{L,H\}$ with prior $\mu(L)=\mu(H)=0.5$, parameters
$(b_L,\lambda_L)=(1,0.1)$, $(b_H,\lambda_H)=(2.4,0.5)$, and cost $c=2$.  
Designer welfare is convex in coverage:
\[
V(a,\theta)=\alpha_\theta \Bigl(\tfrac{n(a)}{N}\Bigr)^{\beta},
\qquad \beta=1.5,
\]
with $\alpha_L=6$ and $\alpha_H=12$.

\paragraph{Constructive policy.} Full-vaccination potentials are $F_L(N)=-2.85$ and $F_H(N)=1.95$, giving scores
$S(L)=-0.475$ and $S(H)=0.1625$, hence $H \succ L$.
Solving the sequential-obedience balance
\[
\mu(H)F_H(N)+p^\star\mu(L)F_L(N)=0
\]
yields $p^\star=0.684$ with threshold $\theta^\star=L$.

The optimal policy $\pi^{\mathrm{seq}\,\star}$ is:
\begin{itemize}
\item in $H$, recommend vaccination to all agents;
\item in $L$, recommend vaccination to all agents with probability $0.684$, and to none otherwise.
\end{itemize}
When recommending vaccination, probability is split equally across all six orderings.
This achieves expected welfare $8.05$.

\paragraph{\textbf{Why private sequential signaling matters.}} If the same randomization ($0.684$ in $L$, $1$ in $H$) were implemented via a public signal, observing “vaccinate’’ yields
\[
\Pr(H\mid\text{vac})
=
\frac{0.5}{0.5+0.5*0.684}
\approx 0.595.
\]
They would then play a simultaneous-move game under this common posterior, where the expected marginal gain from vaccinating is
\[
G_i(a_{-i},\text{posterior})
= 0.595\,G_i(a_{-i},H) + 0.405\,G_i(a_{-i},L).
\]
In particular, if no one vaccinates, each agent’s deviation payoff is
\[
G_i(\mathbf{0},\text{posterior})
= 0.595 * 0.4 + 0.405 * (-1)
= -0.167 < 0,
\]
so universal non-vaccination is a Bayesian NE and the \emph{smallest}-equilibrium of the simultaneous game. Under public signaling, the intended vaccination policy collapses to no adoption. The robust optimum, therefore, relies on private sequential invitations.

\paragraph{\textbf{Contrast with partial-implementation (optimistic) design.}} A designer checking only classical obedience would recommend vaccination in both states, predicting welfare $9$.  
Since the signal is uninformative, agents retain the prior $0.5$, yielding
\[
G_i(a_{-i},0.5)=0.5\,G_i(a_{-i},H)+0.5\,G_i(a_{-i},L).
\]
Then, we can verify that both \emph{no vaccination} and \emph{full vaccination} are Bayesian NE:
\begin{itemize}
    \item \textbf{No vaccination is a Bayesian NE:}  
    \[
    G_i(\mathbf{0},0.5)=0.5*(0.4)+0.5*(-1)=-0.3<0,
    \]
    so no agent wishes to vaccinate when everyone declines.
    \item \textbf{Full vaccination is a Bayesian NE:}  
    \[
    G_i(\mathbf{1},0.5)=0.5*(0.4+0.5)+0.5*(-1+0.1)=0\ge 0,
    \]
    so no agent has an incentive to deviate from vaccination when all others vaccinate.
\end{itemize}

Thus, under partial implementation, this policy is \emph{optimal} because full vaccination is a Bayesian NE in every state. Yet under smallest-equilibrium play, agents coordinate on no vaccination, yielding welfare 0. This illustrates why sequential obedience is essential: ignoring it can produce policies that appear optimal but collapse under conservative reasoning. 

\paragraph{\textbf{Take-home message.}} The threshold rule (vaccinate in $H$, mix in $L$ with $p^\star=0.684$) is simple, exact, and robust.  
Sequential obedience is essential: public or optimistic designs fail once agents reason conservatively.

\subsection{Case study~2: Technology adoption}

\paragraph{\textbf{Setup.}}
The vaccination example illustrated how the robust policy works in a simple two-state setting and why smallest-equilibrium reasoning matters.  We now consider a richer environment that approximates a continuous type space: a technology adoption problem with \(N=10\) agents and nearly continuous uncertainty.  
States are \(\theta \in \{0.01,0.02,\ldots,1.00\}\), representing \emph{ecosystem readiness}.  
As readiness rises, both the private benefit \(b_\theta\) (from \(0.5\) to \(2\)) and the complementarity \(\lambda_\theta\) (from \(0.1\) to \(0.8\)) increase linearly.  
This fine discretization approximates a continuous domain within our finite-state framework. 
Designer welfare uses the same convex coverage function as in Case Study~1, with \(\alpha_\theta\) increasing from \(6\) to \(12\).

\paragraph{\textbf{Baselines and evaluation.}}
We compare three policies:
\begin{enumerate}[label=(\arabic*)]
    \item \textbf{Robust (ours).} The policy \(\pi^{\mathrm{seq}\,\star}\), which satisfies (SO-C) and (SO-N) and is robust under smallest-equilibrium play.  
    \item \textbf{BCE-Optimistic.} Solves the designer’s problem under classical obedience (Bayesian persuasion/BCE), assuming coordination on the designer-preferred equilibrium.
    \item \textbf{BCE-Realized.} Takes the BCE policy from (2) but \emph{evaluates it} under smallest-equilibrium play.
\end{enumerate}

\paragraph{\textbf{Results.}}
With \(c=2\), Fig.~\ref{fig:case2_combined} shows that the score 
\(S(\theta)=F_\theta(N)/\allowbreak V(\mathbf{1},\theta)\) orders the states smoothly, yielding an almost continuous threshold.  
The robust policy prescribes full adoption at high readiness, none at low readiness, and mixes at \(\theta^\star=0.56\) with \(p^\star=0.24\).  
In contrast, BCE-Optimistic recommends adoption for all \(\theta>0.28\) but collapses to no adoption under smallest-equilibrium play.

Fig.~\ref{fig:case2_welfare} compares welfare under the three policies: 
\textbf{Robust} (blue), \textbf{BCE-Optimistic} (orange), and \textbf{BCE-Realized} (red).  
The orange curve shows welfare \emph{predicted by classical BCE analysis}, assuming agents coordinate on the designer-preferred equilibrium.  
The red curve shows the \emph{realized welfare} when agents instead play the smallest-equilibrium: whenever “no adoption’’ is also a BNE 
(\(1.5\!\le\!c\!\le\!2.95\)), agents select that equilibrium, and welfare drops to zero. 
Hence, BCE-Optimistic overstates achievable welfare, while BCE-Realized reveals what the same policy delivers with cautious agents.  
The robust policy (blue) achieves the true coordination frontier under smallest-equilibrium play.

At low costs ($c \le 1.5$), all policies coincide (universal adoption); at high costs ($c \ge 2.95$), all converge to zero (universal non-adoption).  
Overall, the results highlight the practical value of smallest-equilibrium robustness: it identifies credible, implementable coordination outcomes where classical persuasion fails.

\begin{figure}[t]
    \centering
    \includegraphics[width=0.98\linewidth]{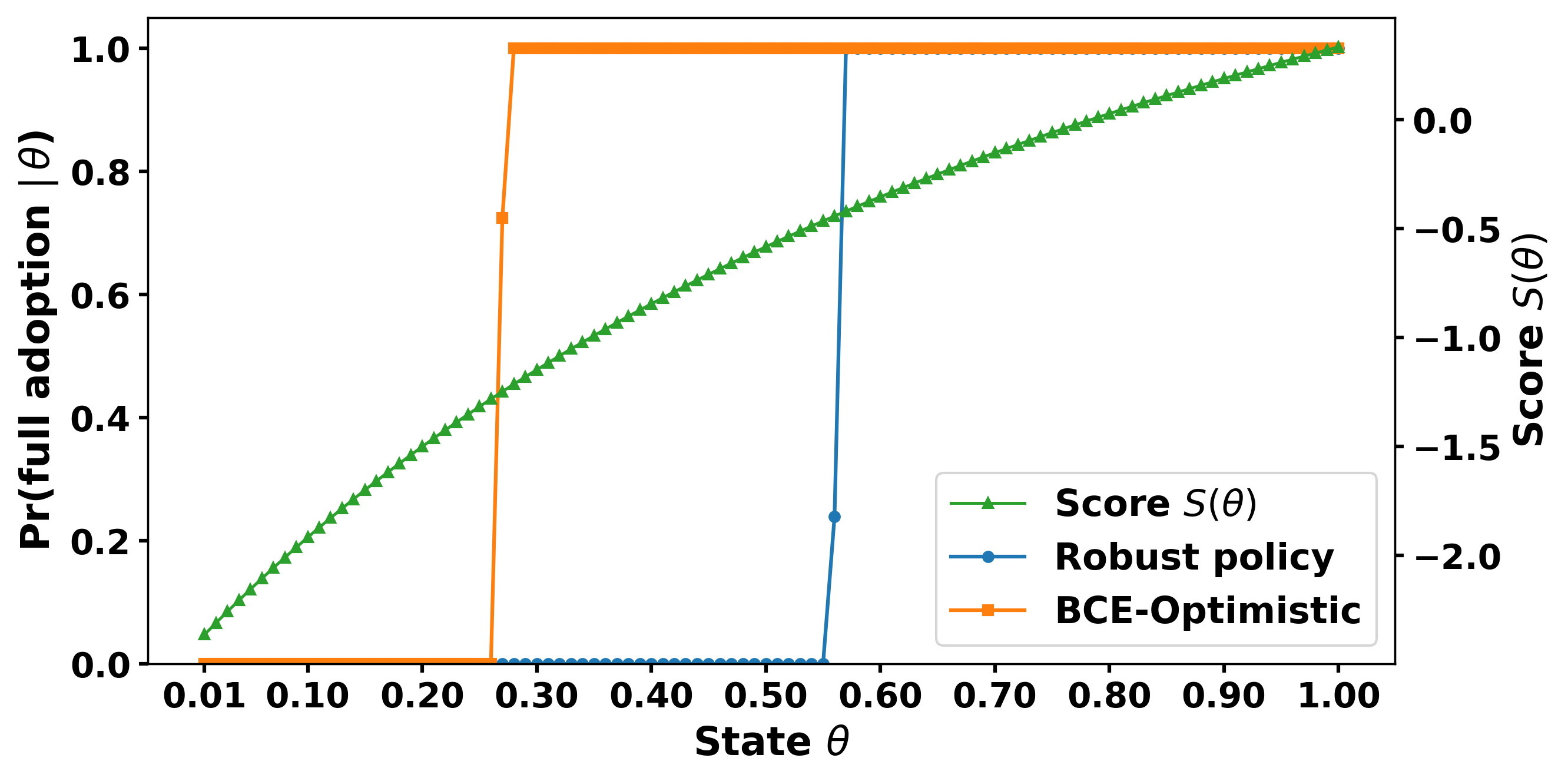}
    \caption{Technology adoption with nearly continuous states. 
Scores \(S(\theta)\) induce a near-continuous threshold. 
The robust policy \(\pi^{\mathrm{seq}\,\star}\) adopts only where coordination is sustainable, while BCE policy (orange) recommends broader adoption but collapses to no adoption under smallest-equilibrium play.}
    \label{fig:case2_combined}
\end{figure}
\begin{figure}[t]
    \centering
    \includegraphics[width=0.98\linewidth]{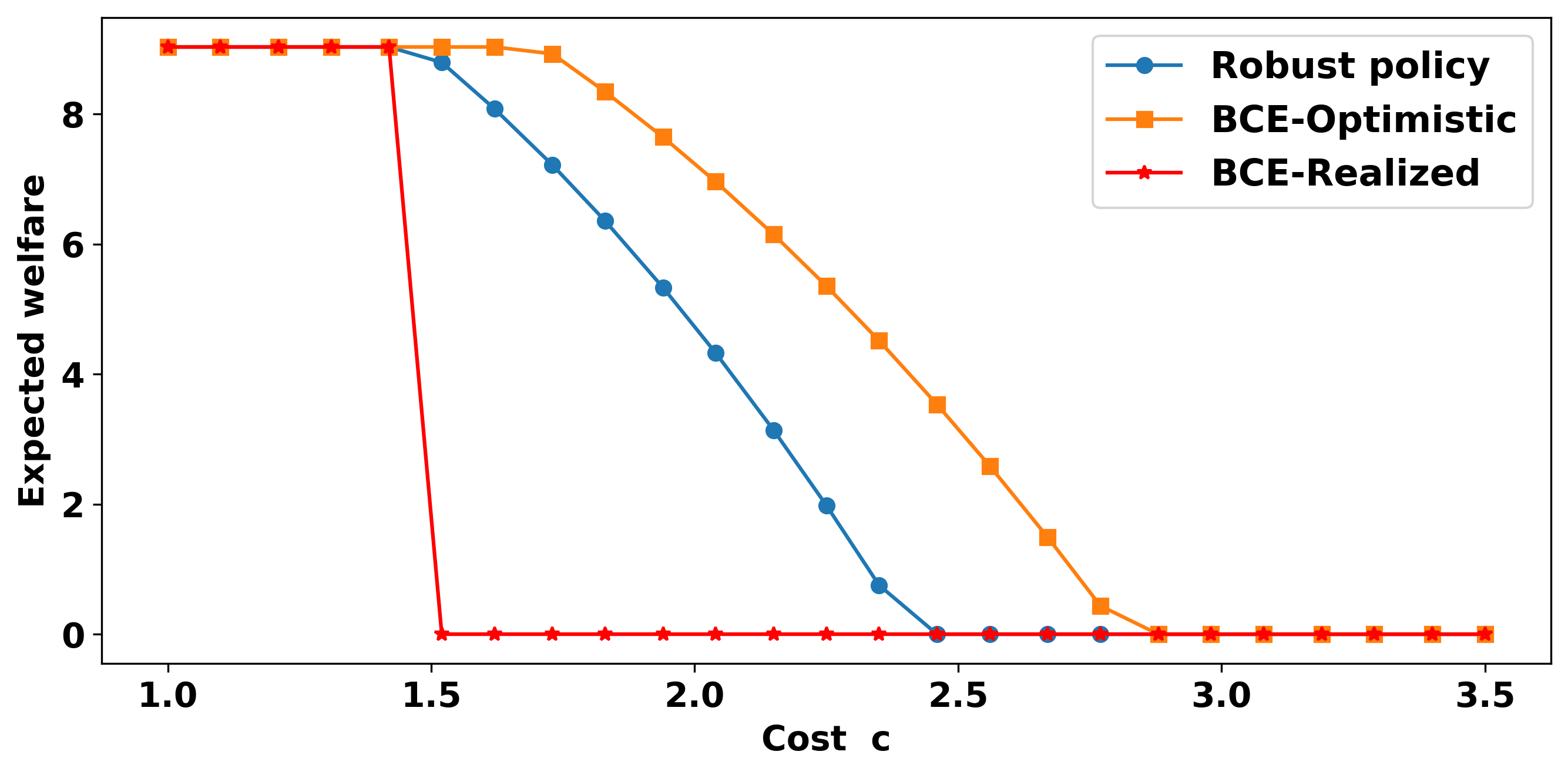}
    \caption{Welfare under Robust, \textbf{BCE-Optimistic} (classical obedience only), 
    and \textbf{BCE-Realized} (same BCE policy evaluated under smallest-equilibrium). 
    For medium costs, the BCE-Optimistic curve sits above Robust, but \textbf{BCE-Realized} drops to \textbf{zero}, 
    revealing that the optimistic gains are not achievable under conservative play.}
    \label{fig:case2_welfare}
\end{figure}

\paragraph{Discussion.}
This case shows the constructive threshold rule extends beyond small finite settings.  
Even though our analysis assumes finitely many states, the technology-adoption example demonstrates that with fine discretization the method handles nearly continuous domains, yielding stable, interpretable policies and welfare.  
Together, the two case studies illustrate both the rule’s clarity in simple environments and its robustness in richer ones.

\section{Conclusion and future work}

We develop a constructive threshold framework for robust information design under smallest-equilibrium play. Unlike classical obedience-based approaches, our method ensures robustness to coordination failure. The resulting policies are simple, interpretable, and scalable to nearly continuous domains while capturing the true limits of sustainable coordination.

\paragraph{\textbf{Future work.}}
One direction is to study agents embedded in social networks, where complementarities are local rather than global, and to analyze how network topology affects optimal disclosure.  
Another extension is to apply the framework to real vaccination data, enabling empirical validation of smallest-equilibrium effects and threshold-based policies.
Finally, one could shift the designer’s objective from maximizing compliance to promoting diversity (e.g., avoiding monopolies); in this case the perfect-coordination reduction no longer applies, but the potential-based scoring approach may be adapted to reward heterogeneous outcomes.



\section{Acknowledgement}
This paper has been accepted for publication in Proceedings of the 25th International Conference on Autonomous Agents and Multiagent Systems (AAMAS 2026). The final published version is available via the ACM Digital Library.

\bibliographystyle{ACM-Reference-Format} 
\bibliography{sample}


\end{document}